\documentclass[
pre,
reprint,superscriptaddress
]{revtex4-1}
\usepackage{bbm}
\usepackage{xcolor}
\usepackage{enumitem}
\usepackage{amsmath}
\usepackage{amssymb}
\usepackage{amsthm}
\usepackage{graphicx}
\usepackage{hyperref}
\usepackage{cleveref}

\IfFileExists{todonotes.sty}{%
  \usepackage[bordercolor=black,backgroundcolor=red!10!white]{todonotes}%
}{%

}

\newtheorem{theorem}{Theorem}
\crefname{lemma}{lemma}{lemmas}
\newtheorem{lemma}{Lemma}
\newtheorem{corollary}{Corollary}
\newtheorem{definition}{Definition}

\let\originalleft\left
  \let\originalright\right
\renewcommand{\left}{\mathopen{}\mathclose\bgroup\originalleft}
  \renewcommand{\right}{\aftergroup\egroup\originalright}

\newcommand{\pj}{p_{j}}
\newcommand{\mco}{\mathcal{O}}
\newcommand{\oj}{\mco_j}

\newcommand{\de}{d_{\mathrm{eff}}}
\newcommand{\id}{\mathbbm{1}}

\newcommand{\expo}[1]{\operatorname{e}^{#1}}
\newcommand{\tr}[2][ ]{\operatorname{Tr}#1\left[ {#2} \right]} 

\newcommand{\braket}[1]{\vphantom{\left(#1\right)^A} \left \langle #1 \right \rangle }

\newcommand{\ket}[1]{\left|#1 \right \rangle \vphantom{\left( #1 \right)^A}}
\newcommand{\bra}[1]{\left\langle #1 \right | \vphantom{\left(#1\right)^A}}

\newcommand{\dif}{\;\mathrm{d}}

\newcommand{\Real}{\mathbb{R}}
\newcommand{\Nat}{\mathbb{N}}

\newcommand{\hil}{\mathcal{H}}

\newcommand{\rank}[1]{\relax\ifmmode\operatorname{rank}#1\else rank-$#1$\fi}
\newcommand{\Set}[2]{\left\lbrace #1 \,\middle|\, #2 \right \rbrace}

\newcommand{\nth}[1][n]{$#1$\textsuperscript{th}{} }

\allowdisplaybreaks[4]

\hypersetup{pdfborder={0 0 1}, bookmarks=true, linkcolor=blue}
\newcommand{\mbm}{\mathbb{M}}
\newcommand{\mcm}{\mathcal{M}}
\newcommand{\mcs}{\mathcal{S}}
\newcommand{\mcp}{\mathcal{P}}

\begin{document}

\title{Comparing classical and quantum equilibration}

\author{Artur S.L. Malabarba}
\affiliation{H.H. Wills Physics Laboratory, University of Bristol, Tyndall Avenue, Bristol, BS8 1TL, U.K.}
\author{Terry Farrelly}
\affiliation{Institut f{\"u}r Theoretische Physik, Leibniz Universit{\"a}t, Appelstra{\ss}e 2, 30167 Hannover, Germany}
\author{Anthony J. Short}
\affiliation{H.H. Wills Physics Laboratory, University of Bristol, Tyndall Avenue, Bristol, BS8 1TL, U.K.}
\date{\today}

\date{\today}

\begin{abstract}
    By using a physically-relevant and theory independent definition of measurement-based equilibration, we show quantitatively that equilibration is easier for quantum systems than for classical systems, in the situation where the initial state of the system is completely known (pure state).
    This shows that quantum equilibration is a fundamental, nigh unavoidable, aspect of physical systems, while classical equilibration relies on experimental ignorance.
    When the state is not completely known, a mixed state, this framework also shows quantum equilibration requires weaker conditions.
\end{abstract}
\maketitle

For over a century, physicists have sought to understand the emergence of apparently irreversible equilibration from reversible microscopic dynamics.
Results over the last few years have shown equilibration for general quantum systems under very weak assumptions\cite{Lin09, Reimann10, Short11, ShortFarrelly11, Goldstein13, Malabarba14, EFG15, GHRdRS16, GE16}.
In classical mechanics, results on equilibration rely on assumptions such as ergodicity, mixing dynamics, mixed initial states, and coarse-graining of the measurements\cite{ReiEvs13, farquhar1964ergodic, gemmer2009quantum, sklar1995physics}.

Recently, Reimann and Evstigneev\cite{ReiEvs13} have analysed equilibration in classical and quantum theory by evaluating observable expectation values, and comparing them to a parameter representing experimental precision.
Then, considering mixed initial states, they are able to compare the conditions necessary for equilibration under each case, showing that they require a very different set of conditions.

Our work complements and extends the work in \cite{ReiEvs13}, using a theory independent definition of equilibration to compare the two scenarios, for both pure and mixed initial states.
Following \cite{Short11, ShortFarrelly11, Malabarba14}, our definition employs a generalized distinguishability which incorporates the full outcome distribution of a measurement, corresponding to its ability to help us distinguish a time-evolving state from a time-invariant equilibrium state.

On the classical side, we show that pure systems equilibrate only when the measurement coarse-graining partitions the state-space in a very specific (and very uneven) way.
On the other hand, quantum equilibration needs only that the measurement be coarse-grained at all, with no restriction on how it partitions the state-space.
Then, using a parameter to characterize measurements on both theories, we are able to show that the range of values which permit classical equilibration is very close to the range which guarantees equilibration on any theory.
Meanwhile, quantum equilibration is possible for a much wider range of this parameter.

Furthermore, when the initial state of the system is taken to be mixed, our approach corroborates the previous results by Reimann and Evstigneev\cite{ReiEvs13}.

\section{Definitions}
\label{sec:definitions}

Although we will specifically consider quantum and classical dynamics, we start our discussion by talking about equilibration and distinguishability strictly in terms of measurements and outcome probabilities, which could be applied to almost any theory.
As such, we need to restate some of our definitions in general terms, without relying on quantum observables or density matrices.

We consider the possible states of a system to be elements in a compact convex space $\mcs$ (herein, a state space).  When we refer to functions on $\mcs$ being linear, this is with respect to convex mixtures in this space.

We compare different states using only the information provided by measurements, as defined below.
\begin{definition}[Measurement]
    \label{def:1}
    Given a state space $\mcs$ and $N\in\Nat$, a measurement with $N$ outcomes is any set of $N$ linear functions $\mcm = \Set{\pj}{j=1,\ldots,N}$, where
    \begin{align*}
      \pj \colon \mcs &\to [0,1]\\
      \rho &\mapsto \pj(\rho)
    \end{align*}
    satisfy $\sum_{j = 1}^{N} \pj(\rho) = 1,\, \forall \rho \in \mcs$.
    Each of these functions represent the probability of obtaining the \nth[j] outcome when measuring $\mcm$ on $\rho$.
\end{definition}
The maximum information that can be gathered about a state, using measurements, is the probability of each particular outcome.

\begin{definition}[Distinguishability]
    \label{def:2}
    Given two possible states of a system, $\rho$ and $\sigma \in \mcs$, and a measurement $\mcm$ with $N$ outcomes.
    The distinguishability between $\rho$ and $\sigma$ according to $\mcm$ is
    \begin{equation}
        \label{eq:8}
        D_\mcm(\rho,\sigma) =
        \frac{1}{2} \sum_{j=1}^N \left| \pj(\rho) - \pj(\sigma) \right|,
    \end{equation}
    where $D_\mcm:\mcs\times\mcs \to [0,1]$.
\end{definition}

The distinguishability is defined this way so that, after performing the measurement, the distinguishability quantifies the average probability of successfully ``guessing'' which state the system was in~\cite{Short11}, according to
\begin{equation}
    \label{eq:11}
    p_{\text{success}} = \frac{1}{2} + \frac{1}{2}D_\mathcal{M}(\sigma,\rho).
\end{equation}
When $D_\mathcal{M}(\sigma,\rho) = 0$ the measurement does not provide information that helps to distinguish $\sigma$ from $\rho$.
On the other hand, when $D_\mathcal{M}(\sigma,\rho) = 1$ the states are perfectly discriminated by this measurement.
This provides a physical and practical meaning to the distinguishability, i.e. if $D_\mathcal{M}(\sigma,\rho) \approx 0$ then $\sigma$ and $\rho$ are experimentally indistinguishable.

In order to talk about equilibration, we also need a notion of time evolution.
We describe it via a function on $\mcs \times \Real_{\geq 0}$, which is linear on $\mcs$ and whose value represents the state $\rho$ evolved by some time $t \geq 0$,
\begin{align*}
  T:\mcs\times\Real_{\geq 0}&\to\mcs \\
  (\rho, t) &\mapsto T_{t}(\rho),
\end{align*}
where $T_0(\rho)=\rho$.
For short, we'll abbreviate $T_t(\rho)$ as $\rho_t$.
Note that in general $T$ does not need to be time-independent or reversible.
However we will be particularly interested in reversible evolutions, for which there exists a function $T^{-1}$ (defined in the same way as $T$) such that $T^{-1}_t (T_t (\rho)) = T_t (T^{-1}_t (\rho)) = \rho$.
We'll use the assumption of reversilibity when we talk about classical mechanics, but general time evolution is enough to define equilibration and even prove our first theorem.

\begin{definition}[Equilibration]
    \label{def:3}
    Given a state $\rho$, a time-evolution $T$, and $0 \leq  \epsilon < 1$, we say that $\rho$ equilibrates up to $\epsilon$ (or $\epsilon$-equilibrates) under this evolution with respect to a measurement $\mcm$ when both averages $\omega = \braket{\rho_t}$ and $\braket{D_\mcm(\rho_t,\omega) }$ exist and satisfy
    \begin{equation}
        \label{eq:28}
        \braket{D_\mcm(\rho_t,\omega) } \leq \epsilon.
    \end{equation}
    Here, $\braket{\cdot} = \lim_{T \rightarrow \infty}\frac{1}{T} \int_{0}^{T} \cdot \dif t$ denotes the time average.
\end{definition}
This definition applies regardless of the dynamics which govern the evolution of $\rho_t$, i.e.\ it is theory independent, so we must account for the possibility of the averages not existing.
Fortunately, in the particular cases of interest to us the dynamics guarantee the convergence and existence of the averages.
Since we only consider compact state spaces here, in quantum mechanics the time average equals a decoherence in a finite-dimensional energy basis, and in classical mechanics the averages converge by Birkhoff's Theorem\cite{farquhar1964ergodic}. Note that in what follows we will not discuss the timescale for equilibration, which may be very long \cite{ShortFarrelly11, Malabarba14, Goldstein13}.

When $\omega = \braket{\rho_t}$ does exist, we call it the equilibrium state.
Since the probabilities are linear functions, the probabilities on $\omega$ can be written as $\pj(\omega) = \braket{\pj(\rho_t)}$. \footnote{Note that strictly speaking, we could replace the condition that $\omega$ exists with the slightly weaker condition that \unexpanded{$ \langle p_{j}(\rho_t) \rangle$}  exist for all $j$, but we use $\omega$ for convenience and to link with previous literature.}.
In addition, we have that $\braket{\rho_t} = \braket{\rho_{t+\tau}}$ for all $\tau \in \Real$, and so the equilibration of $\rho$ also implies the equilibration of $\rho_t$ for all $t$.

Furthermore, note how it is only required that $\rho$ be close to $\omega$ for most times, and not for all large-enough times.
This condition is much weaker and more physically meaningful, firstly because it doesn't preclude the possibility of recurrence in the time evolution (which is possible in Hamiltonian mechanics), and secondly because it has been shown experimentally that systems do fluctuate away from equilibrium~\cite{Donald1962, Wang02}.
This approach is well established in the field of quantum equilibration\cite{Goldsteinnew, Malabarba14, Goldstein13, ShortFarrelly11, Brandao12, Short11, Lin10, Lin09}, and is also being used to describe classical equilibration~\cite{ReiEvs13, Werndl15}.

In particular, Werndl and Frigg\cite{Werndl15} define that a macroscopic state (a region of the state space $\mcs$) is an $\alpha$-$\delta$-equilibrium state if the fraction of time spent inside it is $\geq \alpha$ for all states in $\mcs$ except for a fraction $\delta$ of them.\footnote{They call it $\alpha$-$\epsilon$-equilibrium, but we've replaced $\epsilon$ with $\delta$ here to avoid confusion with our own $\epsilon$ which actually corresponds to their $\alpha$.}
This definition is more oriented towards the macroscopic aspect, a top-down approach that defines an equilibrium macro-state in terms how much time the micro-states spend inside it.
Meanwhile we define equilibration in terms of the initial micro-state and the measurement probabilities, a bottom-up approach.
Still, the two are related in their definition of equilibration in terms of ``most times''.

Finally, it is also useful to define
\begin{definition}[Pure and mixed states]
    \label{sec:pure-states}\label{def:5}
    A state $\psi\in\mcs$ is pure if and only if it cannot be written as a convex sum of other states in $\mcs$ (i.e. the pure states are the extreme points of  $\mcs$).
    We denote as $\mcp$ the set of all pure states in $\mcs$. A state is mixed if it is not pure.

\end{definition}

We now show that reversible time evolutions must preserve purity---i.e., $\psi \in \mcp$ implies $\psi_t \in \mcp \, \forall t$.
This is because otherwise one could find a $\psi\in\mcp$ such that
\begin{align}
  \label{eq:1}
  T_t( \psi) &= q \rho + (1-q) \rho'
\end{align}
where $\rho, \rho' \in \mathcal{S}$ with $\rho \neq \rho'$, and $0< q<1$. Applying $T^{-1}$ to both sides of this equation we find
\begin{align}
  \psi &=q T^{-1}_t (\rho) + (1-q)  T^{-1}_t (\rho')
\end{align}
which means that $T^{-1}_{t} (\rho) = \psi$ and $T^{-1}_{t} (\rho') = \psi$. However, this would imply  $\rho = T_t (\psi) = \rho'$ which contradicts the assumption that $\rho \neq \rho'$.


\section{General Equilibration}
\label{sec:gener-equil}

Here, we show that a very uneven measurement coarse-graining (with respect to the state space explored by the system) is a sufficient condition for equilibration in any theory where the average $\braket{\rho_t}$ exists.
As explained above, this includes quantum and classical mechanics.

\begin{theorem}[Sufficiency]
    \label{thr:4}
    Take any $\rho \in \mcs$, any time evolution $\rho_t$ such that $\omega = \braket{\rho_t}$  and $\braket{D_\mcm(\rho_t,\omega) }$ exist, and any measurement $\mcm$.
    $\rho_t$ $\epsilon$-equilibrates under $\mcm$ if
    \begin{equation}
        \label{eq:31}
        \max_{j}\pj(\omega)  \geq 1- \frac{\epsilon}{2}.
    \end{equation}
\end{theorem}
\begin{proof}
    \label{pr:1}
    First, without loss of generality we label as $1$ the outcome that satisfies \cref{eq:31}, then we note that $p_1(\rho_t) - p_1(\omega) \leq 1 - p_1(\omega) \leq \frac{\epsilon}{2}$.
    This leads to
    \begin{align}
      \label{eq:29}
      \braket{\left| p_1(\rho_t) - p_1(\omega) \right|}
      &= \braket{p_1(\rho_t) - p_1(\omega) }^+ + \braket{p_1(\rho_t) - p_1(\omega) }^- \notag\\
      &= 2\braket{p_1(\rho_t) - p_1(\omega) }^+ \leq {\epsilon},
    \end{align}
    where $\braket{f(t)}^\pm = \braket{\max\{{\pm f(t), 0}\}}$ and we use the fact that $\braket{f(t)}^+ = \braket{f(t)}^-$ whenever $\braket{f(t)} = 0$.

    One also has
    \begin{equation}
        \label{eq:10}
        \sum_{j = 2}^N \pj(\omega) = 1 - p_1(\omega) \leq \frac{\epsilon}{2},
    \end{equation}
    which, in turn, leads to
    \begin{align}
      \label{eq:3333}
      \braket{D_\mcm(\rho_t,\omega) }
      &= \frac{1}{2} \braket{\left| p_1(\rho_t) - p_1(\omega) \right| }\! \notag\\
      &\quad\quad\quad\quad + \frac{1}{2}\! \!\sum_{j = 2}^N \braket{\left| \pj(\rho_t) - \pj(\omega) \right| }\notag\\
      &\leq \frac{\epsilon}{2} + \frac{1}{2} \sum_{j = 2}^N \braket{\pj(\rho_t) + \pj(\omega)}\\
      &= \frac{\epsilon}{2} + \sum_{j = 2}^N \pj(\omega) \leq \epsilon. \notag
    \end{align}
\end{proof}
This result says that one will always observe equilibration if the measurement being used is bad enough, i.e., if one of the outcomes is predominantly more likely than all the others most of the time.
We will herein refer to these as highly uneven measurements.

It is important to understand that this applies to any state, pure or mixed, of any theory that matches the definitions above, be it quantum, classical or otherwise.
In other words, if a measurement satisfies this assumption for a given state and a given time evolution, it is so bad at distinguishing the time evolving state from the equilibrium state that one is guaranteed to have equilibration regardless of any other properties of the theory.

Below we take a similar approach to study what conditions are \emph{necessary} for equilibration under each theory.

\subsection{Classical Equilibration}
\label{sec:classical-1}

Both in classical and quantum mechanics, mixed states represent a lack of knowledge regarding the parameters of the system.
Thus, we start by studying the case where the initial state is pure, so any subsequent equilibration is strictly a consequence of the theory and not of previous ignorance.

Below, we show a necessary condition for classical equilibration of pure states which is very similar to the sufficient condition above.
Which means classical pure states only equilibrate when the measurement is very bad in a very specific way.
In contrast, for quantum mechanics, we provide a counter example showing that the same condition is not necessary.

In order to define classical dynamics for our purposes, we only need three of its properties.
The first property, is that time-evolution is reversible.
The second defining characteristic is that for pure states, at any point in time, the outcome of any measurement is completely determined. The third is that time-averages exist.
\begin{definition}[Classical Mechanics]
    \label{def:4}
    A given combination of state space $\mcs$, $N$-outcomes measurement $\mcm$, and time evolution $T$,  obey classical mechanics only if $T$ is reversible, the averages $\omega = \braket{\rho_t}$ and $\braket{D_\mcm(\rho_t,\omega) }$ exist for any initial state, and  $\exists\, \xi:\mcp\to \{1,2,\ldots N\}$ such that
    \begin{equation}
        \label{eq:5}
        \pj(\psi_t) = \delta_{j,\xi(\psi_t)},\, \forall \psi \in \mcp,
    \end{equation}
    with $j = 1, \ldots, N$ and $\pj \in \mcm$.
\end{definition}
Consequently, one has $\pj(\braket{\psi_t}) = \braket{\delta_{j,\xi(\psi_t)}}$.
Of course, binary measurement probabilities are not all that defines classical mechanics, there are many properties (specially on the time evolution) that are not being taken into account here.
However, since the theorem below is a necessity theorem showing how hard equilibration is, adding further constraints to our definitions cannot make equilibration any easier.

\subsubsection{Classical Equilibration of Pure States}

In words, the following theorem then says that a classical pure state will only equilibrate with respect to $\mcm$ if the evolving state spends most of its time inside the support of a single outcome of $\mcm$.
\begin{theorem}[Classical Necessity]
    \label{thr:1}
    A classical pure state $\psi$ may $\epsilon$-equilibrate with respect to $\mcm$ only if
    \begin{equation}
        \label{eq:2}
        \max_j \braket{\pj(\psi_t)} = \max_j \pj(\omega) \ge 1- \epsilon,
    \end{equation}
    where $\omega = \braket{\psi_t}$.
\end{theorem}
\begin{proof}
    First, one has that $\forall \pj \in \mcm$
    \begin{align}
      \label{eq:3}
      \left| \pj(\psi_t) - \pj(\omega) \right|
      &= (1 - \pj(\omega))\delta_{j,\xi(\psi_t)} \notag\\
      &\qquad\quad+ \pj(\omega) (1- \delta_{j,\xi(\psi_t)}). \notag\\
      \Rightarrow \braket{\left| \pj(\psi_t) - \pj(\omega) \right|}
      &= 2 \pj(\omega) \left[ 1 - \pj(\omega) \right].
    \end{align}
    where in the second step we have used $\braket{\delta_{j,\xi(\psi_t)}} = \braket{\pj(\psi_t)} =\pj(\omega)$. This implies
    \begin{align}
      \label{eq:4}
      \braket{D_\mcm(\psi_t,\omega) }
      &= \sum_{j = 1}^N \pj(\omega) \left[ 1 - \pj(\omega) \right] \notag\\[-.2cm]
      &= 1 - \sum_{j = 1}^N \pj(\omega)^2.
    \end{align}
    The $\epsilon$-equilibration condition is then written as
    \begin{align}
      & 1 - \sum_{j = 1}^N \pj(\omega)^2 < \epsilon \notag\\[-.3cm]
      \Rightarrow & 1 - \epsilon < \sum_{j = 1}^N \pj(\omega)^2 < \max_{j}\pj(\omega)
    \end{align}
\end{proof}
Note how similar the inequality in \cref{eq:2} is to \cref{eq:31}.
The Sufficiency \cref{thr:4} shows how uneven a measurement needs to be so that even the most stubborn of systems must equilibrate under it.
The Necessity \cref{thr:1} shows that any classical measurement which allows pure states to $\epsilon$-equilibrate is at most an $\epsilon/2$ away from being one of these exceptionally uneven measurements.
That is, pure states in classical mechanics are among the hardest of all systems to equilibrate.

\subsubsection{Classical Equilibration of Mixed States}
\label{sec:mixed-classical-case}

In the case of Hamiltonian dynamics, the theorem below is a statement that equilibration will be achieved when (i) there is a chaotic subspace of $\mcp$, and (ii) the initial state can be represented as a mixture of states mostly within this chaotic subspace.
This theorem is an extension of the results in \cite{ReiEvs13}, applied to the distinguishability as defined in \cref{eq:8}.
\begin{theorem}
    \label{thr:3}
    A classical mixed state $\rho$ will $\epsilon$-equilibrate with respect to $\mcm$ if there exists a $\mcp_c \subseteq \mcp$ such that:
    \begin{enumerate}
      \item Two different time-parametrized states in $\mcp_c$ are uncorrelated, when averaging over all time.
        That is,
        \begin{equation}
            \label{eq:25}
            \braket{\pj(\psi_t)\pj(\psi'_t)} = \pj(\braket{\psi_t})\pj(\braket{\psi'_t}),
        \end{equation}
        for any $\pj\in \mcm$ and almost all pairs $(\psi, \psi') \in \mcp_c\times \mcp_c$.

      \item The state $\rho$ can be described as a mixture of pure states mostly contained in $\mcp_c$, i.e.
        \begin{equation}
            \rho = \int_{\mcp}\mu(\psi) \psi  \dif\psi.
        \end{equation}
        where $\mu(\psi)$ is a non-negative function satisfying $\int \mu(\psi) \dif\psi =1$ \footnote{Note that by taking $\mu(\psi)$ to be a function, we exlcude probability distributions involving delta functions, which could yield different results because  $\mathcal{I}_j^{\mcp_c \times \mcp_c}$ could be non-zero in that case.}, such that
        \begin{equation}
            \label{eq:26}
            \int_{\mcp_c} \mu(\psi) \dif\psi \geq 1- \delta,
        \end{equation}
        where $\delta =2 \epsilon^2/N \leq \frac{1}{2}$.
    \end{enumerate}
\end{theorem}

Note that $\psi$ and $\psi'$ can be very close.
The reason $\mcp_c$ is called the chaotic subspace is that after a large enough time, even an infinitesimal difference between these two states must become large enough make their probabilities be uncorrelated.

The following proof is similar to calculations by Reimann and Evstigneev\cite{ReiEvs13}, with the difference that we use the distinguishability instead of measurement expectation values.
\begin{proof}[Proof of \cref{thr:3}]
    Firstly, we note that, for any $\rho$ (mixed or pure) and any time evolution,
    \begin{align}
      \label{eq:14}
      \braket{D_\mcm(\omega,\rho_t)}
      &= \frac{1}{2} \sum_{j = 1}^{N} \braket{\left| \pj(\rho_t) - \pj(\omega) \right| } \notag\\
      &= \frac{1}{2} \sum_{j = 1}^{N} \braket{\sqrt{\left[ \pj(\rho_t) - \pj(\omega) \right]^2} }\notag\\
      &\leq \frac{1}{2} \sum_{j = 1}^{N} \sqrt{\braket{\left[ \pj(\rho_t) - \pj(\omega) \right]^2}} \notag\\
      &\leq \frac{1}{2} \sqrt{N \sum_{j = 1}^{N} \braket{\left[ \pj(\rho_t) - \pj(\omega) \right]^2}} \notag\\
      &= \frac{1}{2} \sqrt{N \sum_{j = 1}^{N} \braket{\pj(\rho_t)^2} - \pj(\omega)^2}.
    \end{align}
    Applying this equation to the current scenario, we have
    \begin{align}
      \label{eq:19}
      \frac{4}{N}\braket{D_\mcm(\omega,\rho_t)}^2
      &= \sum_{j = 1}^{N} \braket{\pj(\rho_t)^2} - \pj(\omega)^2 \notag\\
      &= \sum_{j = 1}^{N} \mathcal{I}_j^{\mcp \times \mcp} \notag\\[-.2cm]
      &= \sum_{j = 1}^{N} \left( \mathcal{I}_j^{\mcp_c \times \mcp_c} + \mathcal{I}_j^{\mcp_c \times \mcp_p} \right.\notag\\
      &\quad\quad\quad\quad\left. + \mathcal{I}_j^{\mcp_p \times \mcp_c} + \mathcal{I}_j^{\mcp_p \times \mcp_p} \right), \notag
    \end{align}
    where $\omega = \braket{\rho_t}$, $\mcp_p = \mcp \setminus \mcp_c$ is called the periodic subspace and
    \begin{equation}
        \label{eq:27}
        \hspace{-0.25cm}
        \mathcal{I}_j^R = \!\!\iint_{R}\Big[
        \braket{\pj(\psi_t)\pj(\psi'_t)} - \pj(\braket{\psi_t})\pj(\braket{\psi'_t})
        \Big] \mu(\psi)\mu(\psi') \; \dif \psi \, \dif \psi'.
    \end{equation}
    From \cref{eq:25}, it is clear that
    \begin{equation}
        \label{eq:20}
        \sum_{j = 1}^{N} \mathcal{I}_j^{\mcp_c \times \mcp_c} = 0.
    \end{equation}
    From \cref{eq:26}, combined with the identity
    \begin{align}
      \label{eq:59}
      \sum_{j=1}^N\Big[ \braket{\pj(\psi_t)\pj(\psi'_t)} - \pj(\braket{\psi_t})\pj(\braket{\psi'_t}) \Big]
      &\leq \sum_{j=1}^N\braket{\pj(\psi_t)\pj(\psi'_t)} \nonumber \\
      & \leq \sum_{j=1}^N\braket{\pj(\psi_t)} \nonumber \\
      &\leq 1
    \end{align}
    and given that we have assumed $\delta \leq \frac{1}{2}$, we have
    \begin{align}
      \label{eq:21}
      \sum_{j = 1}^{N} \mathcal{I}_j^{\mcp_p \times \mcp_p} &\leq \delta^ 2, \notag\\
      \sum_{j = 1}^{N} \mathcal{I}_j^{\mcp_p \times \mcp_c} = \sum_{j = 1}^{N} I_j^{\mcp_c \times \mcp_p} &\leq \delta(1-\delta).
    \end{align}
    Which finally gives
    \begin{equation}
        \label{eq:22}
        \braket{D_\mcm(\omega,\rho^\mu_t)} \leq \sqrt{\frac{N \delta}{2}} \leq \epsilon.
    \end{equation}
\end{proof}

\subsection{Quantum Equilibration}
\label{sec:quantum-1}

Quantum states are represented by density matrices acting on a Hilbert space.
The measurement $\mcm$ is defined in terms of a set of positive operators $\oj$ acting on the same space, each corresponding to an outcome of the observable and satisfying $\sum_{j = 1}^{N} \oj = \id$.
The probabilities are then defined by $\pj(\rho) = \tr{\oj\rho}$.

In the quantum case, the condition of \cref{thr:4} is analogous to saying that one of the operators $\oj$ is close to the identity on almost all of the accessible state space (i.e. over the energy states which occur in $\rho$).
However, while this condition is sufficient, it is not at all necessary and even relatively fine-grained measurements will lead to equilibration.

As such, the following theorem is not a necessity statement like \cref{thr:1} in the classical case.
Instead, we provide an alternative sufficiency theorem which applies to measurements not encompassed by \cref{thr:4} (those for which all $\pj(\omega)$ are very small).

Remarkably, as has been well investigated\cite{Goldstein14, Malabarba14, Masanes13, Goldstein13, Campos12, Brandao12, ShortFarrelly11, Short11, Lin09}, we can obtain the same bounds on quantum equilibration for both pure and mixed states.
Here we provide an improved version of a bound known from the literature.
\begin{theorem}
    \label{thr:5}
    The average distinguishability between a general quantum state $\rho\in\mcs$ evolving via a static Hamiltonian,  and its time average with respect to an $N$-outcome measurement $\mcm$ satisfies
    \begin{equation}
        \label{eq:6}
        \braket{D_{\mcm}(\rho_t,\omega)} \leq \frac{1}{2} \sqrt{D_G \frac{N - 1}{\de}}.
    \end{equation}
    where $D_G$ is the maximum degeneracy among gaps in the system's spectrum and $\de^{-1} = \sum_{n} \tr{\rho\Pi_n}^2 $ is the effective dimension, with $\Pi_n$ being the projector onto the $n$-th energy eigenspace.
\end{theorem}
The proof of this theorem is a fairly straightforward improvement on a previous proof from the literature \cite{ShortFarrelly11}, so we present it in \cref{sec:quantum-th}.
Then, it is easy to apply this result to the current definition of $\epsilon$-equilibration.

\begin{corollary}
    \label{thr:2}
    Under the conditions of \cref{thr:5}, a general quantum state $\rho$ will $\epsilon$-equilibrate with respect to an $N$-outcome measurement $\mcm$ if
    \begin{equation}
        \label{eq:9}
        N \leq 4 \frac{\de}{D_G} \epsilon^2 + 1,
    \end{equation}
    where $D_G$ and $\de$ are defined as in  \cref{thr:5}.
\end{corollary}

The effective dimension, $\de$, is a recurring parameter in the field of quantum equilibration.
It roughly quantifies how many energy eigenstates a system occupies with significant probability and is usually assumed to be very large.

Considering the effective dimension typically scales exponentially in the number of particles, it easily outgrows the number of outcomes of any conceivable physical measurement.
In this case the system will equilibrate no matter how the outcomes partition the Hilbert space.

Here, we see that the quantum conditions of a system with a small number of degenerate energy gaps and a state with high effective dimension achieve a similar effect as the classical condition of chaoticity of the initial state.
In fact, apart from some small constants, the fraction $\delta$ of the classical distribution which lies outside the chaotic subspace is under the same restriction as $D_G/\de$, which suggests these two quantities are analogues of each other.

\subsection{Multiple Measurements}
\label{sec:n-measurements}

Instead of defining equilibration in terms of the distinguishability with respect to a single measurement, one can also take a set of $K$ measurements and consider the maximum distinguishability amongst all of them at any point in time.
This approach\cite{ShortFarrelly11} describes the scenario where one is capable of performing different measurements on the system, and always knows which one would be best at distinguishing the state at a given time from the equilibrium state (or has multiple copies of the system and performs all measurements).

In the language of this paper, we consider a set of $K$ measurements $\mbm=\Set{\mcm_i}{i=1,\ldots,K}$ and define the distinguishability with respect to this set of measurements as
\begin{equation}
    \label{eq:12}
    {D_{\mbm}(\rho_t,\omega)} = \max_{\mcm\in\mbm} {D_\mcm(\rho_t,\omega)}.
\end{equation}
Given this change in the definition, it is useful to consider how each of the above theorems would have to be adapted to account for multiple measurements.

\Cref{thr:4,thr:3,thr:2} still hold if we replace $\epsilon$ with $\epsilon/K$, where their respective assumptions must hold for all $\mcm\in\mbm$.
Then we have that
\begin{align}
  \label{eq:15}
  \braket{{D_{\mbm}(\rho_t,\omega)}}
  &= \braket{\max_{\mcm\in\mbm} {D_\mcm(\rho_t,\omega)}\!\!} \notag\\
  &\leq\!\! \sum_{\mcm\in\mbm}\!\! \braket{{D_\mcm(\rho_t,\omega)}} \leq K \frac{\epsilon}{K}.
\end{align}

\Cref{thr:1} still holds if all $\mcm\in\mbm$ satisfy \cref{eq:2}, but the equation itself need not be changed.

\section{Discussion}
\label{sec:discussion}

Here, we have used the distinguishability to define an operational, theory-independent, notion of equilibration with respect to a given measurement. Although we have applied this to quantum and classical theory, it could be applied in other cases such as general probabilistic theories \cite{barrett07}.

We first show a sufficient condition for equilibration in any theory, which depends on the largest average probability among the measurement outcomes. We find that a value of $1-\frac{\epsilon}{2}$ for this parameter is a sufficient condition for $\epsilon$-equilibration for any theory that fits our definitions---which simply means one always observes equilibration if the measurement being used is uneven enough.

In order to achieve equilibration of pure states under classical Hamiltonian dynamics (simply by virtue of how classical measurements are defined), it is necessary that this parameter be at least $1- \epsilon$, showing that classical mechanics is very similar to the worst possible case.

In contrast, we have also shown an alternative sufficient condition for quantum equilibration, which shows that it can happen even when this parameter is very close to $0$.
This quantitatively shows that equilibration of pure states is easier under quantum dynamics, at least where it pertains to measurement ignorance. Indeed,  the quantum case seems closer to the classical mixed state case, where our results corroborate other recent results\cite{ReiEvs13}.

While it is difficult, if not impossible, to experimentally prepare large systems in pure states, these pure-state results are very important because they show that equilibration of microscopic systems is a fundamental aspect of physics, not just a consequence of ignorance.
In particular the results imply it is much harder to keep quantum systems out of equilibrium than one might think, even for something as small as a nanoscale system.
After all, even if one takes every possible action to reliably prepare it in a pure state and reliably isolate it from the environment, the measurement used to examine it still needs to have a number of outcomes exponential in the number of particles in the system, otherwise equilibration will be inevitable (assuming there aren't very many degenerate energy gaps).

\acknowledgements
AJS acknowledges support from the Royal Society and FQXi through SVCF\@.
ASLM acknowledges support from the CNPq.
TF is grateful for support from the ERC grants QFTCMPS and SIQS, and by the cluster of excellence EXC201 Quantum Engineering and Space-Time Research.

\bibliographystyle{unsrt}

\appendix

\section{\Cref{thr:5}}
\label{sec:quantum-th}

For clarity, we start by proving a lemma for pure states, and then use that to prove \cref{thr:5} for general states.
As mentioned in the main text, this result is an improvement over previous bounds in the literature by a factor of $\sqrt{N}$.
The step responsible for this improvement is \cref{eq:7}.

\begin{lemma}
    \label{le:4}
    Given a finite-dimensional Hilbert space $\hil$, a projector $P$, and a pure initial state $\rho:\hil \rightarrow \hil$ evolving under a Hamiltonian with energy levels $E_n$, then there exists an energy basis $\{\ket{n}\}$ such that
    \begin{equation}
        \label{ap-eq:35}
        \braket{\tr{P(\rho_t-\omega)}^2} \leq
        \sum_{n \neq j} \sum_{k \neq l} v_{nj} v_{kl}^* \delta_{G_{nj},G_{kl}}.
    \end{equation}
    where $G_{nj} = E_n - E_j$, $\omega = \braket{\rho_t}$, and $v_{nj} = \rho_{nj}P_{jn}$, in which $\rho_{nj} = \braket{n|\rho|j}$, $P_{nj} = \braket{n|P|j}$.
\end{lemma}

\begin{proof}
    \label{pr:2}
    First note that, since $\rho$ is a pure state, even if some energy levels are degenerate, there is always an energy basis we can choose such that $\rho$ only has support on one energy state inside each energy eigenspace.
    In this basis, $G_{nj} = 0$ implies that either $n=j$ or $\rho_{nj} = 0$.

    Then, let us write $\rho(t) = \sum_{nj} \rho_{nj} \expo{-i G_{nj} t} \ket{n}\!\bra{j}$, and note that
    \begin{align}
      \label{ap-eq:36}
      \omega &=  \sum_{n=j} \rho_{nj} \ket{n}\!\bra{j} \notag\\
      \rho(t) - \omega &= \sum_{n \neq j} \rho_{nj} \expo{-i G_{nj} t} \ket{n}\!\bra{j}.
    \end{align}
    Thus, for any projector $P$ we have
    \begin{align}
      \langle& \left| \tr{P (\rho(t)-\omega)} \right|^2 \rangle \notag\\
             &= \Big\langle \Big| \sum_{n \ne j} P_{jn} \rho_{nj} \expo{-i G_{nj} t} \Big|^2 \Big\rangle \\
             &= \sum_{n \neq j} \sum_{k \neq l} \rho_{nj}P_{jn}(\rho_{kl}P_{lk})^*
               \braket{\expo{-i(G_{nj} - G_{kl})t}}, \notag\\
             &= \sum_{n \neq j} \sum_{k \neq l} \rho_{nj}P_{jn}(\rho_{kl}P_{lk})^*
               \delta_{G_{nj},G_{kl}}.\notag
    \end{align}
\end{proof}

\begin{theorem}
    \label{le:3}
    Given a finite-dimensional Hilbert space $\hil$, any initial state $\rho:\hil \rightarrow \hil$ evolving under a Hamiltonian with energy levels $E_n$, and any $N$-outcome measurement $\mcm$, the following equation holds:
    \begin{equation}
        \label{ap-eq:3}
        \big\langle D_{\mcm}(\rho_t,\omega) \big\rangle
        \leq \frac{1}{2} \sqrt{D_G \frac{N - 1}{\de}},
    \end{equation}
    where $D_G$ is the maximum degeneracy of any gap in the spectrum.
\end{theorem}
\begin{proof}
    \label{pr:7}
    We start by assuming that $\rho$ is a pure state, and then generalize to mixed states.
    This means we can use \cref{le:4}.

    Next, the following steps are easier to follow if we label each energy gap by the indices $\alpha = (n,j)$ and $\beta = (k,l)$, denote summing over $\alpha$ as shorthand for summing over $n\neq j$, and we define the Hermitian matrix $M_{\alpha\beta}= \delta_{G_{nj},G_{kl}}$.

    Using these definitions it is easy to see that the sums in \cref{le:4} form an inner product, $\sum_{\alpha \beta} v_\alpha^* M_{\alpha \beta} v_\beta = |v^\dag M v|$.
    Therefore, we can use the Cauchy-Schwarz inequality to bound this sum by
    \begin{align}
      \label{ap-eq:37}
      \braket{\tr{P(\rho_t-\omega)}^2}
      &\leq \sum_{\alpha \beta} v_\alpha^* M_{\alpha \beta} v_\beta \notag\\
      &\le \|M\| \sum_{\alpha} |v_\alpha|^2 \notag\\
      &= \|M\| \sum_{n \neq j} |\rho_{nj}P_{jn}|^2 \notag\\
      &\le \|M\| \sum_{n, j} |\rho_{nj}P_{jn}|^2 \notag\\
      &\le \|M\| \sum_{n, j} |\rho_{jj}P_{jn} \rho_{nn} P_{nj}| \notag\\
      &= \|M\| \tr{P\omega P\omega},
    \end{align}
    where we also used that $\rho_{jj}\rho_{nn} \geq \rho_{nj}\rho_{jn}$ by positivity of the density matrix.
    Now, note that the left-hand-side doesn't change if you subtract from the projector anything proportional to the identity.
    Therefore, we can write\newpage
    \begin{align}
      \label{eq:7}
      &\braket{\tr{P(\rho_t-\omega)}^2}\frac{1}{\|M\|} \notag\\
      &\quad\le \tr{(P - \frac{\id}{N})\omega (P - \frac{\id}{N})\omega} \notag\\
      &\quad\le \tr{P\omega P\omega} - 2\tr{\frac{\id}{N}\omega P\omega} + \tr{\frac{\id}{N}\omega \frac{\id}{N}\omega} \notag\\
      &\quad= \tr{P\omega P\omega} - \frac{2}{N}\tr{P\omega^2} + \frac{1}{N^2}\tr{\omega^2} \notag\\
      &\quad\leq \tr{P\omega^2}\left(1 - \frac{2}{N}\right) + \frac{1}{N^2}\tr{\omega^2} ,
    \end{align}
    which leads to
    \begin{align}
      \sum_{P \in \mcm}\braket{\tr{P(\rho_t-\omega)}^2}
      &\leq \|M\| \tr{\omega^2}\left( 1 - \frac{2}{N} + \frac{1}{N} \right) \notag\\
      &= \frac{\|M\|}{\de}\frac{N-1}{N},
    \end{align}
    where it was used that $\tr{\omega^2} = {\de}^{-1}$.

    Since $M$ is a block diagonal matrix, where each block is composed of only $1$s and is at most of size $D_G$, then the largest eigenvalue of $M$ is at most $D_G$, and we have $\|M\| \leq D_G$.
    This finally leads to
    \begin{align}
      \label{ap-eq:63}
      \braket{D_\mcm(\rho_t,\omega)}
      &= \frac{1}{2} \sum_{P \in \mcm}\braket{\tr{P(\rho_t-\omega)}} \notag\\
      &\leq \frac{1}{2} \sum_{P\in \mcm} \sqrt{\braket{\tr{P(\rho_t-\omega)}^2}} \notag\\
      &\leq \frac{1}{2} \sqrt{N  \sum_{P\in \mcm} \braket{\tr{P(\rho_t-\omega)}^2}} \notag\\
      &\leq \frac{1}{2} \sqrt{(N - 1) \frac{\|M\|}{\de}} \notag\\
      &\leq \frac{1}{2} \sqrt{D_G \frac{N - 1}{\de}}.
    \end{align}

    \newcommand{\mca}{\mathcal{A}}
    To see that the same will hold for mixed states, take an ancillary Hilbert space $\mca$ with the same dimension as $\hil$ and use it to purify $\rho$.
    That is, define a pure state $\rho'$ on $\hil\otimes\mca$ such that $\tr[_\mca]{\rho'} = \rho$, and define $\omega' = \braket{\rho'_t}$.
    Then it is straightforward to see that
    \begin{align}
      \label{ap-eq:62}
      \braket{D_\mcm(\rho_t,\omega)}
      &= \frac{1}{2} \sum_{P\in \mcm} \braket{\left| \tr[_\hil]{P(\rho_t-\omega)} \right|} \notag\\
      &= \frac{1}{2} \sum_{P\in \mcm} \braket{\left| \tr[_\hil]{P \tr[_\mca]{\rho'_t-\omega'}} \right|} \notag\\
      &= \frac{1}{2} \sum_{P\in \mcm} \braket{\left| \tr{P\otimes\id_\mca ({\rho'_t-\omega'})} \right|} \notag\\
      &= \braket{D_\mcm(\rho'_t,\omega')} \notag\\
      &\leq \frac{1}{2} \sqrt{D_G' \frac{N - 1}{\de(\rho')} },
    \end{align}
    where $D_G'$ is the maximum degeneracy among energy gaps in $\hil\otimes\mca$.

    At last, to reproduce \cref{ap-eq:63}, simply chose a null Hamiltonian for the $\mca$ subspace ($H_\mca = 0$ and $H' = H\otimes\id_\mca$).
    Of course this choice leads to a highly degenerate energy spectrum, but, since $\rho'$ is pure, this doesn't affect any of the quantities by the same argument used at the start of this proof.
    This has the consequence that $D_G' = D_G$, and $\de(\rho') = \de(\rho)$.
\end{proof}
\end{document}